\documentclass[10pt,letterpaper,twocolumn]{article}
\usepackage[utf8x]{inputenc}
\usepackage{ucs}
\usepackage{amsmath,amsthm, url}
\usepackage{amsfonts}
\usepackage{amssymb}
\usepackage{geometry,color,graphicx,url}
  \usepackage{multirow}
\usepackage[english]{babel}

\newtheorem{theorem}{Theorem}
\newtheorem{corollary}[theorem]{Corollary}

\newcommand{\waffect}{\texttt{waffect}}
\def\P{\mathbb{P}}

\title{Alternative Methods for H1 Simulations\\in Genome Wide Association Studies}

\author{V. Perduca \textsuperscript{a,}\footnote{Corresponding authors: Vittorio Perduca, Gregory Nuel, MAP5 Universit\'{e} Paris Descartes,  45 Rue des Saints P\`eres,  75006 Paris, France,  \tt{\{vittorio.perduca, gregory.nuel\}@parisdescartes.fr}} , C. Sinoquet \textsuperscript{b}, R. Mourad \textsuperscript{b,c}, G. Nuel \textsuperscript{a,*} \\ 
\small{\textsuperscript{a} MAP5 - UMR CNRS 8145 Universit\'{e} Paris Descartes, \textsuperscript{b} LINA - UMR CNRS 6241 Universit\'{e} de Nantes,} \\
\small{ \textsuperscript{c} Ecole Polytechnique de l'Universit\'{e} de Nantes}
}

\date{}

\begin{document}

\maketitle

\textbf{To appear in \textit{Human Heredity}.}
\bigskip 

\hrule width 0.15\textwidth height 0.15em

\vspace{-1.3em}

\subsubsection*{Key Words:}

\vspace{-0.7em}

Statistical Power~$\cdot$ Disease Model $\cdot$ Rejection $\cdot$ Markov Chain Monte Carlo (MCMC)~$\cdot$ Backward Sampling $\cdot$ Receiver Operating Characteristic (ROC) $\cdot$ Area Under the Curve (AUC)

\vspace{2em}

\hrule width 0.15\textwidth height 0.15em

\vspace{-1.3em}

\subsubsection*{Abstract:}

\vspace{-0.7em}

{\bf \emph{Objective}:} Assessing the statistical power to detect susceptibility variants plays a critical role in GWA studies both from the prospective and retrospective points of view. 
Power is empirically estimated by simulating phenotypes under a disease model H1. For this purpose, the ``gold'' standard consists in simulating genotypes given the phenotypes ({\it e.g.} Hapgen). We introduce here an alternative approach for simulating phenotypes under H1 that does not require generating new genotypes for each simulation. {\bf \emph{Methods}:} In order to simulate phenotypes with a fixed total number of cases and under a given disease model, we suggest three algorithms: i) a simple rejection algorithm; ii) a numerical Markov Chain Monte-Carlo (MCMC) approach; iii) and an exact and efficient backward sampling algorithm. 
In our study, we validated the three algorithms both on a toy-dataset and by comparing them with Hapgen on a more realistic dataset. As an application, we then conducted a simulation study on a \textit{1000 Genomes Project} dataset consisting of 629 individuals (314 cases) and 8,048 SNPs from Chromosome X. We arbitrarily defined an additive disease model with two susceptibility SNPs and an epistatic effect.
{\bf \emph{Results}:} The three algorithms are consistent, but backward sampling is dramatically faster than the other two. Our approach also gives consistent results with Hapgen. Using our application data, we showed that our limited design requires a biological {\it a priori} to limit the investigated region. We also proved that epistatic effects can play a significant role even when simple marker statistics ({\it e.g.} trend) are used. We finally showed that the overall performance of a GWA study strongly depends on the prevalence of the disease: the larger the prevalence, the better the power.
{\bf \emph{Conclusions}:} Our approach is a valid alternative to Hapgen-type methods; it is not only dramatically faster but also has two main advantages: 1) there is no need for sophisticated genotype models ({\it e.g.} haplotype frequencies, or recombination rates); 2) the choice of the disease model is completely unconstrained (number of SNPs involved, Gene-Environment interactions, hybrid genetic models, etc.). Our three algorithms are available in an ${\tt R}$ package called \waffect\ (``double-u affect'', for weighted affectations).

\subsubsection*{Introduction}

Genome-wide association (GWA) studies are a widely-used approach for localizing susceptibility variants responsible for common complex genetic diseases \cite{Morris}. Such studies involve investigating a huge number of genetics markers such as single nucleotide polymorphisms - SNPs (from hundreds of thousands to millions), for cohorts of cases and controls whose sizes range from thousands to tens of thousands of individuals. GWA studies have met with many successes, most notably for type 1 and type 2 diabetes, inflammatory bowel disease, prostate cancer and breast cancer \cite{Hindorff}.

In GWA studies, very high false positive rates are expected when there is no correction for multiple testing. Symmetrically, control for true negative rate - or power - is necessary. Power estimation is the key to evaluate the efficiency of GWA methods  \cite{Spencer}. The correct estimation of both rates must take into account the existence of high-dependency patterns between SNPs, or linkage disequilibrium (LD). The accurate estimation of the family wise type I error risk in the presence of LD consists in sampling the H0 distribution through permutations of phenotypes \cite{Zhang}. Thus, any association between loci and phenotypes is broken. This permutation strategy is implemented as a gold standard in numerous dedicated packages, together with software suites designed for GWA studies \cite{Aulchenko,Browning08,gonzales07,pollard04,Purcell}.

Power is an even more complicated function of several factors: study design, correlation patterns in the genotypic data, sizes of cohorts, frequency of the susceptibility allele, relative risk conferred by the causal factor, genetic model (additive, dominant, recessive, multiplicative) \cite{Lettre}. 
As a consequence, the analytical computation of power requires simplified assumptions, including the approximation of the test statistic distribution under H1 through a probability law \cite{Klein}. Most power calculators based on analytical approaches are used for two-stage GWA design, e.g. \cite{Menashe,steibel08}. Recently, an analytical approach has been proposed to account for LD, under either H0 or H1 approximation \cite{han09}: a fixed-size sliding window locally accounts for the inter-marker correlation. This approach brings an improvement over block-wise strategies by unifying H0 and H1 processings in the same framework \cite{conneely07,lin05,seaman05}. However, with regards to both accuracy and computational burden, the optimal choice of window size depends on the structure of the data. Moreover, LD blocks are often ambiguous. Thus, the previous sliding window approach cannot account for high-order dependencies between LD blocks. In particular, this method cannot be used to evaluate the power of any novel approach designed to cope with such high order dependencies in GWA studies. In this case, the only solution is to use  computationally intensive simulations.

Similar to sampling under H0, the simulation of the H1 distribution is an appealing way to keep the LD-structured genotypic data. These simulations consist in generating case and control samples which mimic the LD structure in human genome, i.e. in the creation of, say, $k$ datasets under the H1 assumption (at least one SNP is a susceptibility SNP). Nonetheless, breaking any association between a locus (or several loci) and the phenotype is far easier than to introduce such an association in a dataset.

Two main strategies have been developed to simulate H1. (i) The first one \cite{Hyam} consists in first generating a large sample of haplotypes conditional on reference haplotypes such as HapMap haplotypes \cite{Hapmap}. The haplotypes are then paired to build diplotypes. The disease status is affected depending on the penetrance model which involves a susceptibility SNP selected at random. (ii) The second strategy \cite{Spencer} consists in first randomly selecting a disease SNP and generating a fixed number of cases and controls. Then diplotypes are assigned at the disease SNP for cases and controls depending on the penetrance model. Finally, haplotypes are built (two for each diplotype) for all remaining SNPs of the chromosome, conditional on reference haplotypes. Nevertheless, both strategies entail problems of implementation when applied to power estimation in GWA studies. The drawback of the first strategy is that it does not make it possible to control the number of cases and controls. To tackle this issue, rejection sampling of case-control samples can be used, but this leads to a waste of time and data. An illustration of the first strategy is Fregene \cite{Hyam}. The second strategy makes it possible to control these numbers by first fixing them, but requires generating haplotypes for each simulation. The widely-used simulator Hapgen \cite{hapgen} implements the second strategy. 



To this aim, we propose a new method which randomly assigns $n_1$ cases and $n_0$ controls conditional on $n=n_0+n_1$ genotypes and a given disease model. 
Our idea is to affect the status of individual $i$ according to both the probability $\pi_i$ that $i$ is a case and to the constraint that the total number of cases is fixed. The probability $\pi_i$ is computed from the disease model and is proportional to the relative risk for individual $i$.
This method provides several advantages. 
First, it generalizes permutations, which is the gold standard for generating H0 distributions in order to control type I error in multiple testing. Indeed, permutations represent a particular case of our general approach which is obtained by using the same $\pi_i$ for all individuals. Secondly, our method is faster than previous ones \cite{Spencer,peng10} because it does not require a rejection sampling step or the generation of new genotypes for each simulation. Moreover, in the case of LD modeling-based mapping methods, LD pattern identification needs to be performed only once. Thirdly, no assumption is made on the genotype distribution, because the latter is the same for each simulation. The last advantage is that power can be directly assessed using real GWA datasets, such as those provided by the \textit{1000 Genomes Project} \cite{1000gen}, because disease status can be generated according to genotypes without loss due to rejection sampling. This is most valuable when evaluating the performance of new mapping methods. 

The paper is organized as follows. We first describe the method itself, which we called \waffect, then we validate our tool on a toy-example and compare its results to Hapgen \cite{hapgen} simulations. We finally illustrate the interest of our approach with a power study conducted on real GWA data.


\subsubsection*{Materials and Methods}

\vspace{1em}
\noindent{\it Detailed aims and notations}

We consider a GWA study with $n=n_0+n_1$ individuals ($n_0$ controls and $n_1$ cases), and $p$ SNPs. The $i^\text{th}$ phenotype is $Y_i \in \{0,1\}$ ($Y_i=1$ if individual $i$ is a case, $Y_i=0$ if individual $i$ is a control). The $i^\text{th}$ genotype is $X_i \in \{0,1,2\}^p$ (this corresponds to the number of rare alleles for each SNP).

Our disease model, thereafter denoted H1, is defined by the probability $\pi_i=\P(Y_i=1 | X_i)$ that individual $i$ is a case conditional on her/his genotype. In the particular case when all $\pi_i$ have the same non-negative value, the observed genotype has no effect on the phenotype. This case corresponds to the null hypothesis H0. In the particular case when we consider a single SNP ($p=1$), $\pi_i$ takes only three possible values: $\pi_i=f_0$ if $X_i=0$, $\pi_i=f_1=f_0 \text{RR}_1$ if $X_i=1$, and $\pi_i=f_2 = f_0\text{RR}_2$ if $X_i=2$, where $f_0$, $f_1$ and $f_2$ are the penetrances of the disease model, and $\text{RR}_1$ and $\text{RR}_2$ are the relative risks.

It is obviously easy to simulate under H1 using the independent Bernoulli distribution $\mathcal{B}(\pi_i)$ for each individual. Unfortunately, it is quite unlikely that when doing so the case-control design constraint $\mathcal{C}=\{\sum_{i=1}^n Y_i = n_1\}$ will be fulfilled.

For this reason, the standard strategy ({\it i.e.} Hapgen) consists in simulating genotypes conditional on phenotypes using Bayes' formula:
\begin{equation}
\P(X_i | Y_i)=\frac{\P(Y_i | X_i)\P(X_i)}{\sum_{x} \P(Y_i | X_i=x)P(X_i=x)}=\frac{\pi_i \P(X_i)}{\P(Y_i)}. \nonumber
\end{equation}
The advantage of this approach is that constraint $\mathcal{C}$ is automatically fulfilled, whereas the major drawback is that a genotype distribution $\P(X_i)$ must be introduced. In order to be realistic, this imposes sophisticated modeling using haplotype data and recombination rates. 

As an alternative, we suggest several ways to simulate phenotypes under H1 conditional on the genotypes (which therefore remain untouched) while respecting constraint $\mathcal{C}$. To do so, we introduce the notation $Z_j=\sum_{i=1}^j Y_i$ (with $Z_0=0$ by convention) so that the constraint becomes $\mathcal{C}=\{Z_n=n_1\}$.

\vspace{1em}
\noindent{\it Reject}

The first approach is straightforward: 1) draw $(Y_i)_{i=1\ldots n}$ using independent Bernoulli distributions; 2) if constraint $\mathcal{C}$ is fulfilled, retain  $(Y_i)_{i=1\ldots n}$ as a valid sample; if not,  simply reject it and return to 1). The problem with this algorithm is that samples will be rejected at rate $1-\P(\mathcal{C})$. The natural question is then: how (un-)likely is event $\mathcal{C}$?  

$\P(\mathcal{C})$ can be computed recursively by introducing the forward quantities: $F_{i}(m)=\P(Z_i=m)$. Starting with $F_0(m)$ being equal to $0$, except for $F_0(0)=1$, we get the following recursive formula for $1 \leqslant i \leqslant n$:
\begin{equation}
\label{eq:forw}
F_{i}(m)=F_{i-1}(m-1)\pi_i + F_{i-1}(m)(1-\pi_i).
\end{equation}
It follows that the probability of the constraint is $\P(\mathcal C)=P(Z_n=n_1)=F_n(n_1)$.

This approach is easy to understand and implement, however it usually has a very low rate of success. This is the case when the design is large and the $\pi_i$s are low (for instance in the case of a small prevalence). In practice, its use will be limited to small toy-examples.

\vspace{1em}
\noindent{\it MCMC}

In order to overcome this drawback, an interesting idea is to turn to Markov Chain Monte Carlo (MCMC) techniques such as the Metropolis-Hastings class of sampling algorithms \cite{mcmc}. Starting from a configuration of phenotypes fulfilling the constraint ({\it e.g.} with the observed phenotypes), alternate the following two steps: 1) proposal move: select two individuals $i$ and $j$ such that $Y_i=1$ and $Y_j=0$ and then exchange their phenotypes ($Y_i=0$ and $Y_j=1$); 2) accept this move with rate $\alpha$ given by:
\begin{equation}
\alpha=\frac{(1-\pi_i)\pi_j}{\pi_i(1-\pi_j)}. \nonumber
\end{equation}
It can be shown that the sequence of phenotype configurations generated by this algorithm is a Markov chain whose stationary distribution is our targeted constrained distribution H1. In practice, the first configurations which are generated must be discarded since the Markov chain is not yet stationary (this stage is called \emph{burn-in}). In order to obtain a set of independent configurations under H1, one can either repeat the burn-in phase of the algorithm as many times as necessary, or use a sequence of pseudo-independent samples by picking phenotype configurations after the burn-in once every fixed number of iterations.

The main advantage of the MCMC approach is that it will eventually work regardless of the probability of event $\mathcal{C}$. Its drawback is that the number of iterations for burn-in and pseudo-independence necessary for good mixing increases at least linearly with the total number $n$ of individuals. Moreover, the control for stationarity and pseudo-independence is delicate and the practical choice of burn-in and independence thresholds usually requires tedious calibration work.

\vspace{1em}
\noindent{\it Backward sampling}

Our proposed alternative is to turn to exact sampling by introducing the backward quantities $B_i(m)=\P(\mathcal{C} | Z_i=m)$ which can be computed recursively like the forward quantities defined above.

Starting from $B_n(m)$ being equal to $0$, except for $B_n(n_1)=1$, we get the following recursive formula for $1 \leqslant i \leqslant n$:
\begin{equation}\label{eq:form_back}
B_{i-1}(m) = \pi_i B_{i}(m+1)+(1-\pi_{i})B_{i}(m).
\end{equation}
One should note that we obtain $\P(\mathcal{C})=B_0(0)$ in the process.

Then we can sample phenotypes $Y_i$ sequentially with the following probability:
\begin{equation}\label{eq:sampling}
\P(Y_i=1| Z_{i-1}=m, \mathcal{C})=\frac{\pi_i B_i(m+1)}{B_{i-1}(m)}.
\end{equation}
Since this expression depends on $Z_{i-1}$, the corresponding probability obviously depends on the number of cases seen by $i-1$. In mathematical terms, the sequence $(Y_i,Z_i)_{i=1\ldots n}$ is a \emph{heterogeneous} Markov chain.


\vspace{1em}
\noindent{\it H0 simulations}

H0 simulations are simply performed by permutating the phenotypes $Y_i$s. A simple way to do this is to uniformly choose a permutation $\sigma$ of $\{1,\ldots,n\}$ by performing a $O(n \log n)$ quicksort algorithm on a sample of $n$ independent random variables with a uniform distribution on $[0,1]$. For example, this can be achieved through the {\tt sample} command in {\tt R} \cite{R}. Alternatively, one may use \waffect\ with the same probability $\pi_i$ for all individuals.

\vspace{1em}
\noindent{\it Toy-example dataset}

For validation purposes, we first considered a toy-example dataset with $p=1$ SNP. The $n$ genotypes $X_i$ have the following distribution: 80\% have genotype 0, 15\% genotype 1, and 5\% genotype 2. With $n$ being a multiple of 20, these counts are always integers. For our disease model, we considered the following relative risks: $\text{RR}_1=1.5$ and $\text{RR}_2=2.0$ (additive effect of $0.5$). We hence obtained the following penetrances: $f_0$, $f_1=1.5 f_0$ and $f_2=2.0 f_0$. In order to explore the relevance and limits of each of our three approaches, we considered a wide range of values for $n$ and $f_0$.

\vspace{1em}
\noindent{\it Hapgen simulations}

\begin{table}
   $$
 \begin{array}{| c | c | c | c | c | c |}
		\hline
		n & f_0 & \P(\mathcal{C}) & \text{Rejection}  & \text{MCMC} & \text{Backward} \\
		\hline
		\hline
		 20 & 0.2 & 4.5\cdot 10^{-3} & 0.4\text{ s} & 38\text{ s} &  0.05\text{ s} \\
		\hline
		20  & 0.1 & 1.7 \cdot 10^{-5 }  & 1.5\text{ min}  & 38\text{ s} & 0.05\text{ s} \\
		\hline
		20 & 0.07 & 6.7 \cdot 10^{-7} & 38.5\text{ min} & 38\text{ s} & 0.05\text{ s} \\ 
		\hline
		20  & 0.05 & 2.9 \cdot 10^{-8} &  11.2\text{ h} & 38\text{ s} & 0.1\text{ s}\\  
		\hline
		40  & 0.2 &  8.2\cdot 10^{-5} & 17.4\text{ s}  & 1.2\text{ min} & 0.1\text{ s} \\
		\hline
		100 & 0.2 & 8.7\cdot 10^{-10} & \text{NA} & 3.2\text{ min} & 0.2\text{ s} \\ 
		\hline
		100 & 0.1 & 5.8\cdot 10^{-22} & \text{NA}  & 3.2\text{ min} & 0.2\text{ s} \\
		\hline
		100 & 0.01 & 1.1\cdot 10^{-69} & \text{NA} & 3.2\text{ min} & 0.2\text{ s} \\
		\hline
		\end{array}   
$$
   \caption{Running time for generating 100 replicates.}
   \label{tab_time}
 \end{table}


In contrast with simulations under H1, simulations under H0 are expected to provide few statistically significant results. This contrast can be observed through a binary (or two-class) classifier system. Regarding this aspect, we compared the behavior of \waffect\ with the one of the widely used simulator Hapgen \cite{marchini07} under various conditions. The comparison relied on the so-called Receiver Operating Characteristic (ROC) curve, which is a graphical plot of the true positive classification rate (sensitivity) as a function of the false positive rate (specificity), for a continuum of values of a given discrimination criterion. The performance of each simulator is summarized using the area under the ROC curve (AUC): the larger the AUC, the more performant the simulator in generating H1 data contrasting with H0 data (more details about ROC curves can be found below).

Three input parameters must be specified in Hapgen: i) the range of the Minor Allele Frequency (MAF) of the simulated susceptibility SNP; here we considered $[0.2,0.3]$ or $[0.1,0.2]$; ii) the disease prevalence, which we set to $0.01$; iii) and the severity of the disease expressed through the relative risks $\text{RR}_1=f_1/f_0$ for individuals heterozygous at the disease SNP, and $\text{RR}_2=f_2/f_0$ for homozygous individuals (rare allele). We considered a total of $n=629$ subjects ($n_0=315$ controls, $n_1=314$ cases) and $p=9,579$ SNPs in the region delimited by loci 558,390 and 13,930,000 on chromosome 1 (Reference: HapMap files \cite{Hapmap}, CEU population reference, \url{http://hapmap.ncbi.nlm.nih.gov/}). 

For each condition, the penetrances  $f_0$, $f_1$ and $f_2$ used with \waffect\ for comparison purposes were obtained performing several simulations using Hapgen and then averaging the corresponding penetrances. Four datasets were generated. Each was then submitted to a genome-wide association study. In the first dataset (Hapgen - H1 hypothesis), Hapgen provided the genotypes together with the phenotypes. All the other datasets consisted in the aforesaid genotypic data together with phenotypes generated as follows. In the second dataset (Hapgen - H0 hypothesis), the phenotypes under the H0 hypothesis were generated through the method of permutations. In the third dataset (\waffect\ - H1 hypothesis) the phenotypes were affected with \waffect\ . Finally, in the fourth dataset (\waffect\ - H0 hypothesis), we obtained the replicates by affecting the phenotypes after specifying in \waffect\ a uniform probability $\pi_i$ accross all individuals (this is equivalent to permuting the phenotypes). Besides, 1,000 replicates were simulated for each of the four datasets and under each condition (MAF range and relative risks). In the end, 64,000 GWA studies  were performed ($2$ MAF ranges $\times$ $8$ ($\text{RR}_1$, $\text{RR}_2$) pairs $\times$ $4,000$ datasets). Moreover, the comparison of the AUCs obtained for Hapgen and \waffect\ were performed twice, because we used alternatively two distinct statistics of association as the discrimination criterion for the two-class classifier.



\vspace{1em}
\noindent{\it 1000 Genomes dataset}

The dataset consists of the genotypes of $n=629$ individuals ($n_0=315$ controls, $n_1=314$ cases) from the \textit{1000 Genomes Project} \cite{1000gen}. We focused on the first 100,000 SNPs on chromosome X. In the pretreatment stage, we filtered out all the SNPs with Minor Allele Frequency (MAF) less than or equal to 5\%, ending up with $n=8,048$ SNPs. 

We arbitrarily considered a disease model with two interacting disease SNPs. The two SNPs respectively have the base-pair positions 627,641 and 1,986,325, MAFs $0.26$ and $0.23$, and display no linkage disequilibrium ($r^2=0.02$). We considered an additive effect $\beta>0$ and an epistatic effect $\eta$ such that:
\begin{equation}
\pi_i = f_0 \times \left\{
\begin{array}{ll}
1.0 + \beta X_i^1 & \text{if $X_i^2=0$} \\
1.0 + \beta X_i^2 & \text{if $X_i^1=0$} \\
1.0 + \eta + \beta (X_i^1 + X_i^2) & \text{if $X_i^1X_i^2 \neq 0$} \\
\end{array}
\right. \nonumber
\end{equation}
where $X_i^1$ and $X_i^2$ correspond to the genotype of individual $i$ at the two susceptibility SNPs.


%
%

\begin{table}
   $$
 \begin{array}{| c | c | c | c | c |}
		\hline
		n & f_0  & \text{Rejection} & \text{MCMC} & \text{Backward} \\
		\hline
		\hline
		20 & 0.2 & 0.60 & 0.59 & 0.61 \\
		 &  & [0.53, 0.67] & [0.52, 0.66] &  [0.54, 0.68] \\
		\hline
		20 & 0.1 & 0.59 & 0.59 & 0.58 \\
		& &[0.52, 0.66]  & [0.52, 0.66] & [0.51, 0.65]   \\
		\hline
		20 & 0.07 & 0.62 & 0.56 & 0.56 \\
		& &  [0.55, 0.69] & [0.49, 0.63] & [0.49, 0.63] \\
		\hline
		20 & 0.05 & 0.44 & 0.58 & 0.53 \\
		& & [0.37, 0.51]  & [0.51, 0.65] & [0.47, 0.60] \\
		\hline 
		40 & 0.2 & 0.58 & 0.62 & 0.59 \\
		& & [0.50, 0.65]   & [0.54, 0.69] &  [0.52, 0.67]  \\
		\hline 
		100 & 0.2 & \text{NA} & 0.75 & 0.72 \\
		& & & [0.68, 0.81] & [0.65, 0.79] \\
		\hline 
		100 & 0.1 & \text{NA} & 0.64 & 0.68 \\
		& & & [0.57, 0.72] & [0.61, 0.75] \\
		\hline 
		100 & 0.01 & \text{NA} & 0.65 & 0.59 \\
		& & & [0.58, 0.73] & [0.51, 0.67] \\
		\hline
		\end{array}   
$$
   \caption{AUC and 95\% confidence intervals, 100 replicates.}
   \label{tab_auc}
 \end{table}

\vspace{1em}
\noindent{\it Statistics of association}

The association signal of each SNP (single marker analysis) is computed through the PLINK software \cite{plink} using the trend statistic. We denote $p_j$ the p-value of SNP $j$. In order to avoid the multi-testing issue, we considered the following  single real valued statistic: 
\begin{equation}
S_{\rho}=\max_{j \in \mathcal{J}_{\rho} } - \log_{10} p_{j} \nonumber
\end{equation}
where $\mathcal{J}_{\rho}$ denotes the subset of SNPs such that the distance between their loci and the locus of a disease SNP is less than a radius $\rho$ 

For example, $S_\infty$ corresponds to taking the best p-value $p_j$ in negative (decimal) log-scale on the whole dataset. With $\rho=5$ kb, we consider the best p-value of the SNPs within $5$ kb of one of the two susceptibility SNPs. The radius part of the statistic is somewhat unusual since it takes advantage of information which is rarely available in practice: the location of the disease SNPs. However, in our simulation framework, this information helped us discuss the power of the design depending on the {\it a priori} location of the susceptibility SNPs. From a general point of view, starting with no {\it a priori} leads to a choice of  $\rho=\infty$; after a successful linkage analysis, $\rho=100$ kb might be more relevant, and if a candidate gene has been identified, $\rho=5$ kb could be a good choice.

\vspace{1em}
\noindent{\it ROC curves and AUC}

In order to measure the performance of a given statistic for a given design and H1 model, it is natural to compare the sensitivity and specificity of the testing framework for a given threshold of significance. For example, a global threshold of $5\%$ is commonly accepted as a reasonable choice. It might be interesting to avoid doing this arbitrary choice by studying instead the Receiver Operator Characteristic (ROC) curve which is nothing but a graphical representation of the specificities and sensitivities that we can obtain for all possible values of the threshold of significance \cite{Metz}.

\begin{figure}[t]
\centering
\includegraphics[width=0.5\textwidth]{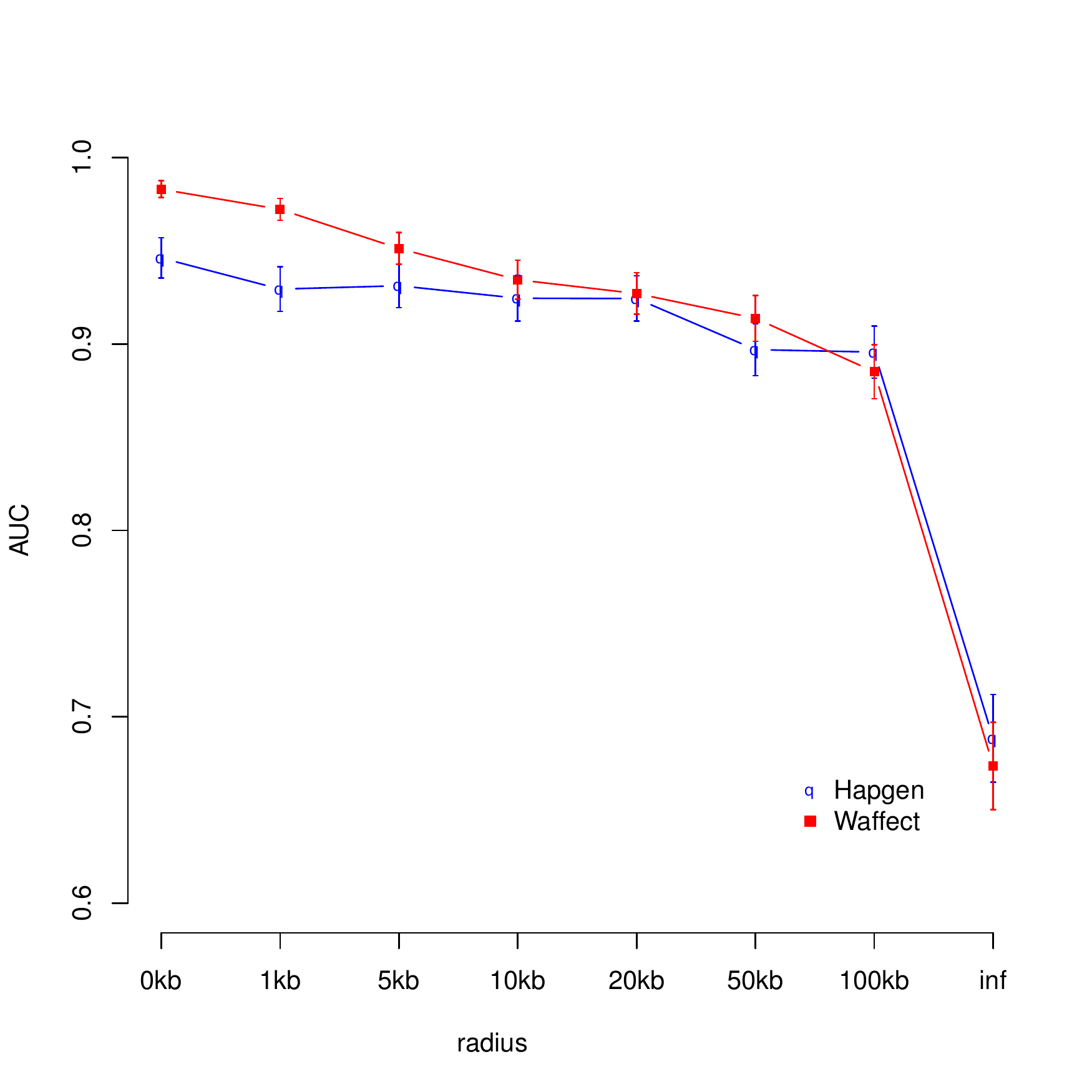}
\vspace{-5mm}
\caption{Comparison of the area under the curves (AUCs) for Hapgen and Waffect - 1,000 replicates under H1 and 1,000 replicates under H0, respectively, for MAF in $[0.2, 0.3]$, additive model, $\text{RR}_1 = 1.6$, $\text{RR}_2 = 2.2$. The 95\% confidence intervals are indicated.}
\label{fig:hapgen_vs_waffect}
\end{figure}

The ROC curve itself can then be further summarized by the Area Under the Curve (AUC) which can be qualitatively interpreted as follows: $\text{AUC} \leqslant 0.6$ means ``fail''; $0.6 < \text{AUC} \leqslant 0.70$ means ``poor''; $0.7 < \text{AUC} \leqslant 0.80$ means ``fair''; $0.8 < \text{AUC} \leqslant 0.9$ means ``good''; $0.9 < \text{AUC} \leqslant 1.0$ means ``excellent''. 

The AUC can easily be estimated from two samples of the statistic: one sample under H0, and the other under H1. All ROC curves and AUC computations (including confidence intervals) were performed using the {\tt R} package {\tt pROC} \cite{pROC}.

\subsubsection*{Results}

\vspace{1em}
\noindent{\it Validation: toy-example dataset} 


In Table \ref{tab_time} are depicted the execution times for generating 100 replicates through the rejection, MCMC and backward sampling algorithms on the toy-example dataset, together with the corresponding probabilities $\P(\mathcal{C})$ of randomly drawing a complete configuration of phenotypes with exactly $n_1$ cases. Different values of $n$ and $f_0$ are considered. $\P(\mathcal{C})$ dramatically decreases when $n$ grows and $f_0$ becomes smaller. As a consequence, the running time of the rejection algorithm is already important (more than 10 hours) even for $n$ as small as $40$ and $f_0$ as large as $0.2$. The rejection sampling algorithm was not run for greater values of $n$ as the corresponding probability of a successful finding is negligible. The MCMC and backward sampling algorithms have short running times which grow  with $n$ (very gently for the backward algorithm), and, most importantly, which do not depend on $\P(\mathcal{C})$. For the MCMC algorithm, we considered a burn-in of $10^5\times n$.

The Area Under the Curves and the corresponding 95\% confidence intervals computed on the toy-example dataset replicates are shown in Table~\ref{tab_auc}. The values found with the three methods are consistent.


\begin{figure}
   \begin{center}
     \setlength{\unitlength}{5mm}
     \includegraphics[scale = 0.5]{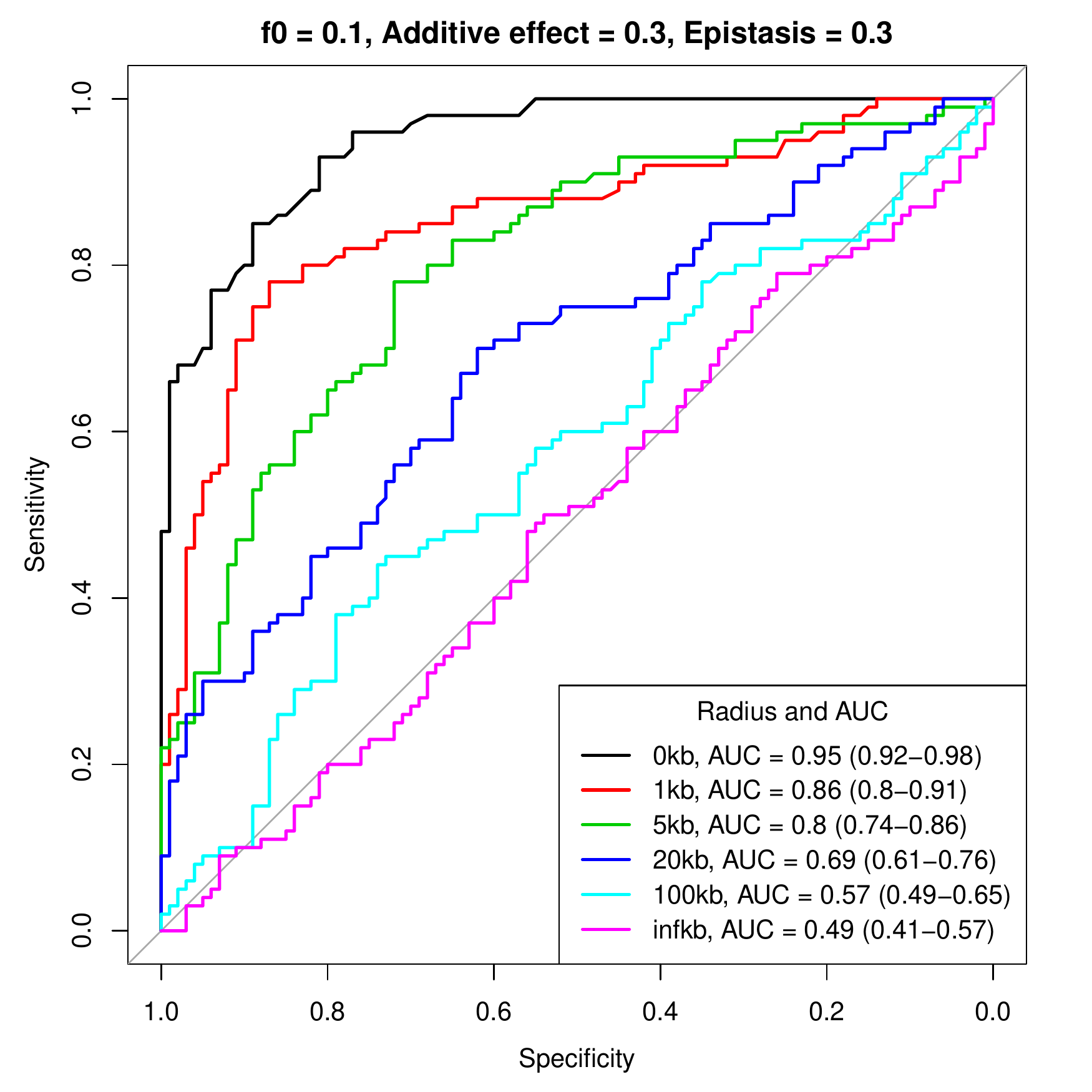}
 \end{center}
   \caption{ROC curves and AUCs at different values $\rho$ of the radius of the intervals centered in the two disease SNPs.}
   \label{fig_reg}
 \end{figure}

\vspace{1em}
\noindent{\it Validation: Hapgen simulations}

First, it has to be noted that even if the number of SNPs specified in Hapgen is constant across the replicates, the observed number of SNPs varies across these replicates. Moreover, the disease SNP's MAF and the relative risks observed in the simulated data fluctuate from one replicate to another (see the Supplementary Material for explanation). 

Figure~\ref{fig:hapgen_vs_waffect}
illustrates the case when the MAF range is $[0.2, 0.3]$ and the genetic model is additive with $\text{RR}_1$ and $\text{RR}_2$ values respectively set to $1.6$ and $2.2$. Despite the fact that Hapgen and our approach are not based on the same exact model, we can see a very good correlation between the two approaches, suggesting that they give similar results.

In our study (629 subjects, 9,579 SNPs specified in Hapgen, Power Edge R900 XEON X7460 2.6 GHz, RAM 128 GB), the generation and processing of 1,000 replicates required 75 minutes or so using the backward sampling algorithm (under H1 or H0) or Hapgen under H0. In contrast, under H1 with Hapgen, 110 hours or so were necessary to generate and process data with the same dimension.

More comprehensive comparisons with Hapgen are available in the Supplementary Material. 


\vspace{1em}
\noindent{\it Application: Power study on the 1000 Genomes dataset}

The ROC curves of the statistics $S_{\rho}$ on the \textit{1000 Genomes Project} dataset for different values of the radius $\rho$ are displayed in Figure~\ref{fig_reg}. The AUC is very low ($95\%$ CI = $[0.41, 0.57]$) when considering all SNPs ($\rho = \infty$). The AUC increases as the region around the two susceptibility SNPs being investigated gets smaller. In particular, the performance is good for $\rho \leqslant 5$ kb. This roughly corresponds to the size of a gene. These results show that given the design and the disease model under consideration, the statistic which consists in simply taking the best p-value is successful in detecting a positive signal only in presence of a biological {\it a priori}, thereby narrowing the investigation to the candidate gene level. 

 %

 %
 
We fixed a candidate region corresponding to the radius $\rho = 5$ kb and investigated the dependency of the AUC on the following parameters: epistasis $\eta$, additive effect $\beta$, probability $f_0$ of being a case in absence of the rare alleles, and total number $n$ of individuals.

 \begin{table}
   \begin{center}
 \begin{tabular}{| c |  c |}
		\hline
		$\beta$ &  AUC $[95\%$ \mbox{ CI}] \\
		\hline
		\hline
		$0.1$ & $0.62$ $[0.55,  0.70]$ \\
		\hline
		$0.2$ & $0.72$ $[0.64, 0.79]$ \\
		\hline
		$0.3$ & $0.80$ $[0.74, 0.86]$ \\
		\hline
		$0.4$ & $0.90$ $[0.86, 0.94]$ \\
		\hline
		\end{tabular}   
 \end{center}
   \caption{AUCs at different values $\beta$ of additive effect. Fixed values: $f_0=0.1$, epistasis $\beta=0.3$, radius of the candidate region $\rho=5$ kb. }
   \label{tab_ad}
 \end{table}

\begin{table}
   \begin{center}
 \begin{tabular}{| c |  c |}
		\hline
		$\eta$ & AUC $[95\%$ \mbox{ CI}] \\
		\hline
		\hline
		$0$ & $0.74$ $[0.67,  0.81]$ \\
		\hline
		$0.3$ & $0.80$ $[0.74, 0.86]$ \\
		\hline
		$0.6$ & $0.85$ $[0.80, 0.91]$ \\
		\hline
		$0.9$ & $0.88$ $[0.83, 0.93]$ \\
		\hline
		\end{tabular}   
 \end{center}
   \caption{AUCs at different values $\eta$ of epistasis. Fixed values: $f_0 = 0.1$, additive effect $\beta=0.3$, radius of the candidate region $\rho=5$ kb. }
   \label{tab_ep}
 \end{table}

Not unexpectedly, the AUC grows with the additive effect $\beta$ (Table~\ref{tab_ad}). More surprisingly, it also grows when the epistatic effect increases (Table~\ref{tab_ep}) despite the fact that our statistic $S_\rho$ only uses simple marker statistics. 
 
 %

 %
 
We then focused on the design of the GWA study itself by studying the effect of the total number of individuals $n$ on the AUC. This was simply done by doubling (tripling, quadrupling) the original dataset before applying our disease model. In the case of our configuration of reference $\eta = 0.3, \beta = 0.3, f_0=0.1$ and $\rho=5$ kb, we observed that the performance is excellent when the population is twice as large as the original one ($\text{AUC} = 0.94\ [0.90,0.97]$ for $n=1258$, $\text{AUC} = 0.80\ [0.74,0.86]$ for $n=629$). This gain in power is even more dramatic when all the SNPs are taken into account ($\rho=\infty$): the AUC doubles in value when passing from $629$ to $1258$ individuals (see Table~\ref{tab_pops}).

Finally, we investigated the effect of $f_0$, the probability of being affected without any exposure, in the design of the GWA study (note that for small values, $f_0$ is close to the prevalence). We can see in Table~\ref{tab_f0} that small values of $f_0$ result in smaller AUCs ({\it i.e.} $\text{AUC}=0.73\ [0.66,0.80]$ for $f_0=0.001$), while large values of $f_0$ result in much better results ({\it e.g.} $\text{AUC}=0.94\ [0.90,0.97]$ for $f=0.25$).


\begin{table}
   \begin{center}
 \begin{tabular}{| c |  c |}
		\hline
		$n$  & AUC $[95\%$ \mbox{ CI}] \\
		\hline
		\hline
		$629$ & $0.49$ $[0.41,  0.57]$ \\
		\hline
		$1258$ & $0.78$ $[0.71, 0.84]$ \\
		\hline
		$1887$ & $0.92$ $[0.88, 0.96]$ \\
		\hline
		$2516$ & $0.93$ $[0.90, 0.97]$ \\
		\hline
		\end{tabular}   
 \end{center}
 \caption{AUCs for different sample sizes and on the whole dataset.  Fixed values: epistasis $\eta=0.3$, additive effect $\beta=0.3$, $f_0=0.1$.}
   \label{tab_pops}
 \end{table}
  
   \begin{table}
   \begin{center}
 \begin{tabular}{| c |  c |}
		\hline
		$f_0$ & AUC $[95\%$ \mbox{ CI}] \\
		\hline
		\hline
		$0.001$ & $0.73$ $[0.66,  0.80]$ \\
		\hline
		$0.01$ & $0.71$ $[0.64, 0.78]$ \\
		\hline
		$0.1$ & $0.80$ $[0.74, 0.86]$ \\
		\hline
		$0.25$ & $0.94$ $[0.90, 0.97]$ \\
		\hline
		\end{tabular}   
 \end{center}
 \caption{AUCs at different values of penetrance $f_0$. Fixed values: epistasis $\eta=0.3$, additive effect $\beta=0.3$, radius of the candidate region $\rho=5$ kb. }
   \label{tab_f0}
 \end{table}

\subsubsection*{Discussion}

\vspace{1em}
\noindent{\it Validation}

The three algorithms we suggested give similar results but their performances are quite different. The rejection algorithm has a complexity of $O( N / \mathbb P(\mathcal C))$ to obtain $N$ samples. This prevents its use in practical cases when $\mathbb{P}(\mathcal{C})$ can easily reach $10^{-50}$ or worse. The MCMC approach has the advantage of removing this dependence, and is simple to implement. The drawback is that the burn-in parameter requires some calibration work (we suggest to use $\text{burn-in}=10^5 n$). Moreover this heavy numerical method is rather slow. In our study, the backward sampling approach proved to be the fastest one, dramatically outperforming the MCMC alternative ({\it i.e.} 500 times faster). From a theoretical point of view, the backward sampling algorithm has a space complexity of $O(n_1 \times n)$, and a time complexity of $O(n_1 \times n + N \times n)$ where $N$ is the desired sample size. Our {\tt C} implementation makes it possible to generate a configuration of $n=20,000$ phenotypes with $n_1=10,000$ cases in $0.2$ seconds on a $1.86$ Gz workstation (Xeon E5320, 8 Go RAM, Linux 2.6.35).

\vspace{1em}
\noindent{\it Comparison with Hapgen}


Our comparison with Hapgen clearly demonstrates that our approach is a valid alternative to this software. It also sheds light on the major differences between the two approaches. First, in order to yield a realistic model for simulating genotypes, Hapgen requires additional information such as HapMap frequencies and recombination rates whereas our approach only needs the original genotypes. Regarding our simulations performed with Hapgen and \waffect\,, running the newest Hapgen version would have entailed no change in the results: the newest Hapgen version is an extension of the former; it is not an improvement (at least regarding single susceptibility SNP simulation). Indeed, in its initial version, Hapgen limits the disease model to only one susceptibility SNP. The most recent version can now simulate multiple disease SNPs, under the assumption that each disease SNP acts independently and that all such SNPs are in Hardy-Weinberg equilibrium, which are two strong constraints. Besides, this novel software, unfortunately only available in the 64 bit version, requires the use of an R package to simulate interactions between the disease SNPs. In contrast, simulating any disease model is straightforward with our approach, as long as it results in an individual probability for each individual to be affected. For example, we can consider two or more susceptibility SNPs, and easily add epistatic effects, Gene-Environment interactions, the contribution of rare variants, etc. Moreover, due to the multiple constraints it must take into account, Hapgen only respects the specified model on average. As a consequence, the effective MAF and relative risks can vary from one run to another. Besides, selecting for each Hapgen replicate a different disease SNP entails variations in the number of SNPs available for the association test (see the Supplementary Material for more details). Finally, since our approach only generates a fraction of the data (the phenotypes), it is outstandingly faster than Hapgen. Furthermore, since with our approach the genotype remains untouched, any numerically intensive analysis which is performed on the genotype table ({\it g.e.} LD computations) needs to be performed only once while with Hapgen it has to be done for each replicate.


\vspace{1em}
\noindent{\it Effects}

In our application study, we showed how our strategy can be used to investigate various effects on the performance of a GWA study. We first pointed out that with the suggested design, our modest disease model ($\beta=\eta=0.3$, $f_0=0.1$), and the analysis we chose (simple marker using PLINK), the signal is only statistically detectable by limiting the region to be investigated to a $5$ kb radius around the susceptibility SNPs. This clearly shows that the present design is practical only if previous work ({\it e.g.} linkage analysis) suggests candidate genes.

Then we proved that the AUC increases when either the additive effect ($\beta$) or the epistatic effect ($\eta$) increases. While the result for the additive effect is not surprising,  the result for the epistatic effect is worth commenting upon. Indeed, in our application, we deliberately stuck to a classical simple marker analysis (PLINK using the trend test) which is not designed to detect epistatic effects. However, our study showed that the corresponding marginal effects can improve the overall statistical performance of the analysis even when no specific effort for the detection of epistasis is performed. It would be interesting to do further research to compare simple marker analysis to more complex analysis especially designed to take into account epistatic effects. 

\vspace{1em}
\noindent{\it Design}

We finally briefly investigated the influence of the design of the GWA study on its performance. As could be expected, the total number $n$ of individuals plays a critical role. It might be interesting however to refine this analysis by also considering different ratios $n_1/n_0$. This is left for further investigation. We also showed that for a fixed design (fixed number of cases and controls) the prevalence of the disease (through $f_0$, which is closely related) has an unexpected effect on the performance: higher prevalence gives better results. More investigations might be necessary to understand the reason behind this observation.

\vspace{1em}
\noindent{\it Extensions}

One should note that our approach can be extended in several directions. First, we can easily consider more than two classes, thus complexifying constraint $\mathcal{C}$. In this case, the rejection and MCMC algorithms remain exactly the same. For backward sampling, one should first affect one class against the others, and then perform the affectation recursively on the remaining unaffected classes. The resulting complexity is $O(n\times (K-1))$ for the affectation of $n$ individuals into $K$ classes. As explained above, we can also consider sophisticated genetic models that take into account additional covariates such as environmental exposure or rare variants. Finally, the extension of our approach to other fields is quite straightforward. For example, in the analysis of gene expression, our approach might be used to affect sample status ({\it e.g.} cancer or healthy) conditional on the level of expression rather than generating expressions through parametric or non-parametric questionable models as is usually done.

\vspace{1em}
\noindent{\it Conclusions}

In this paper, we present an alternative to classical simulations under H1 in GWA studies. The idea is to generate the phenotypes conditional on the existing genotypes with respect to the study design (number of cases). For that purpose, we suggest three algorithms including the backward sampling which mimics Hidden Markov Models to provide a fast sampling conditional on the constraint. Our study shows that our algorithms are valid and that their results are consistent with reference software such as Hapgen. Moreover, our approach has many advantages: it is much faster; it does not need any genotyping model (genotypes remain untouched); it makes it possible to consider any complex genetic model (including several SNPs, epistasis, covariates, rare variants, etc.). Our approaches are available in an ${\tt R}$ package called \waffect\ (``double-u affect'', for weighted affectations). The beta version can be downloaded from \\ \url{www.math-info.univ-paris5.fr/~vperduca} .

\subsubsection*{Authors' contributions}


GN suggested the original idea and the backward sampling algorithm. GN and VP implemented the {\tt R} package. RM and CS wrote the state-of-the-art on the subject in the Introduction. CS took in charge the comparison with Hapgen and the editing of the corresponding supplementary material. VP performed all the other simulations and edited the Appendix. The editing of the paper was done by GN and VP. All authors read and approved the final manuscript. 

\subsubsection*{Appendix}

In this section we give the proofs of the mathematical results from the Material and Methods section, namely of equations (\ref{eq:forw}), (\ref{eq:form_back}) and  (\ref{eq:sampling}). Recall that $Z_j = \sum_{i=1}^j Y_i$ is the number of cases among individuals $1,\ldots, j$ and that we deal with constraint $\mathcal{C}=\{Z_n = n_1\}$. The forward and backward quantities are the probabilities $F_i(m)=\P(Z_i=m)$ and $B_i(m)=\P(\mathcal{C}|Z_i = m)$. Let $V$ be the set of variables $\{Y_1,\ldots,Y_n,Z_1,\ldots,Z_n\}$. 

The key property is that the conditional dependencies among all the variables determine the following factorization of the joint probability distribution
\begin{equation}
\label{eq:factor}
\P(V) = \prod_{j\in I}\P(Z_j|Z_{j-1},Y_j)\P(Y_j),
\end{equation}
where $I=\{1,\ldots,n\}$ is the set of all the individuals and $Z_0=0$ by convention.

\begin{theorem}
For each individual $i=1,\ldots,n$:
\begin{equation} 
\P(Z_i=m,\mathcal{C})  = F_i(m)B_i(m). \label{eq:separator}   
\end{equation}
Moreover 
\begin{multline}
\P(Y_i=1,Z_{i}=m+1,Z_{i-1}=m,\mathcal{C}) = \\
F_{i-1}(m)\pi_iB_i(m+1), \label{eq:clique1} 
\end{multline}
\begin{multline}
\P(Y_i=0,Z_i=m, Z_{i-1}=m,\mathcal{C})   =  \\
F_{i-1}(m)(1-\pi_i)B_i(m)  \label{eq:clique0}.
\end{multline}
\end{theorem}

\begin{proof}
We will only prove Eq. (\ref{eq:separator}): similar arguments hold for equations (\ref{eq:clique1}) and (\ref{eq:clique0}). The marginal probability $\P(Z_i=m,\mathcal{C})$ is obtained from the joint probability distribution $\P(V)$ multiplying it by the indicator function ${\bf 1}_\mathcal{C}$ (by definition ${\bf 1}_{\mathcal{C}}$ = 1 if and only if $\mathcal{C}$ is true) and summing out all the variables in $V\setminus\{Z_i\}=\{Z_1,\ldots,Z_{i-1},Z_{i+1},\ldots,Z_{n},Y_1,\ldots,Y_n\}$. Because of Eq. (\ref{eq:factor}), we have


\begin{equation}
 \label{eq:proof_sep_part}
\P(Z_i=m,\mathcal{C}) = \sum_{V\setminus\{Z_i\}}  {\bf 1}_{\mathcal{C}} \prod_{j\in I} \mathbb{Q}_j(Z_j,Z_{j-1},Y_j),
\end{equation}
where for convenience we denote $\mathbb{Q}_j(Z_j,Z_{j-1},Y_j)=\P(Z_j|Z_{j-1},Y_j)\P(Y_j)$.

Now consider the sets
\begin{align*}
V_{1:i} & :=\{Z_1,Y_1,\ldots,Z_i,Y_i\},  \\
V_{(i+1):n} & :=\{Z_{i+1},Y_{i+1},\ldots,Z_{n},Y_n\}.
\end{align*}
$V_{1:i}\cup V_{(i+1):n}$ is a partition of $V$ and 
$V\setminus\{Z_i\}=V_{1:j}-\{Z_i\}\cup V_{(i+1):n}$.

Note that Eq. (\ref{eq:proof_sep_part}) is equal to 

\begin{align*}
& \sum_{V_{1:i}\setminus\{Z_i\}} \prod_{1\leq j \leq i}\mathbb{Q}_j\left(\sum_{V_{(i+1):n}} {\bf 1}_{\mathcal{C}}\prod_{i+1 \leq j \leq n} \mathbb{Q}_j \right) = \\
&  \left(\sum_{V_{1:i}\setminus\{Z_i\}} \prod_{1\leq j \leq i}\mathbb{Q}_j\right)\cdot\left(\sum_{V_{(i+1):n}} {\bf 1}_{\mathcal{C}}\prod_{i+1 \leq j \leq n} \mathbb{Q}_j\right),
\end{align*}
where the last equality holds because the only variable shared between the two factors is $Z_i$ which is set to be $Z_i=m$ from the beginning. It is now easy to see that the two factors in the product above are $F_i(m)$ and $B_i(m)$ respectively.
\end{proof}

\begin{corollary}
The equations (\ref{eq:clique1}) and (\ref{eq:clique0}) make it possible to compute the forward and backward quantities recursively:
\begin{align*}
F_{i}(m) & =F_{i-1}(m-1)\pi_i + F_{i-1}(m)(1-\pi_i), \\
B_{i-1}(m) & = \pi_i B_{i}(m+1)+(1-\pi_{i})B_{i}(m).
\end{align*}
\end{corollary}

\begin{proof}
We will prove only the recursive formula for $B_i$, a similar argument holds for $F_i$.  Note that 
\begin{multline}
\P(Z_i=m,\mathcal{C}) = F_i(m)\pi_{i+1}B_{i+1}(m+1)  \\ 
+ F_i(m)(1-\pi_{i+1})B_{i+1}(m) = F_i(m)B_i(m) \nonumber
\end{multline}

By dividing by $F_i(m)$, we obtain the iterative formula for the backward quantities. 
%
%
%
\end{proof}

\begin{corollary}
We can now prove Eq. (\ref{eq:sampling}): 
$$
\P(Y_i=1|Z_{i-1}=m,\mathcal{C})=\frac{\pi_iB_i(m+1)}{B_{i-1}(m)},
$$
for each individual $i=1,\ldots,n$.
\end{corollary}

\begin{proof}
Apply the definition of conditional probability and divide Eq. (\ref{eq:clique1}) by Eq. (\ref{eq:separator}).
\end{proof}

\begin{corollary}
\begin{align*}
\P(Y_i = 1 | \mathcal{C}) & \propto   \sum_m F_i(m)\pi_iB_i(m+1),\\
\P(Y_1 = 0 | \mathcal{C}) & \propto \sum_m F_{i-1}(m)(1-\pi_i)B_i(m).
\end{align*}
\end{corollary}

\bibliographystyle{unsrt}
\bibliography{waffect}

\end{document}